\documentclass[11pt,a4paper]{article}
\usepackage[latin9]{inputenc}
\usepackage[english]{babel}
\usepackage[margin=1in]{geometry}

\usepackage{amsthm}

\usepackage{amsmath}
\usepackage{amsfonts}
\usepackage{amssymb}
\usepackage{graphicx}
\usepackage{bbm}
\usepackage{color}
\usepackage{xspace}
\usepackage{hyperref}
\usepackage{enumitem}
\usepackage{tcolorbox}
\usepackage{cleveref}
\usepackage{authblk}

\newtheorem{thm}{Theorem}%[section]
\newtheorem{defi}[thm]{Definition}
\newtheorem{lemma}[thm]{Lemma}
\newtheorem{conj}{Conjecture}
\newtheorem{result}{Result}
\newtheorem{feat}{Property}

\newcommand{\ket}[1]{| #1 \rangle}
\newcommand{\bra}[1]{\langle #1 |}
\newcommand{\braket}[2]{\langle #1 | #2 \rangle }
\newcommand{\proj}[1]{\ket{#1}\bra{#1}}
\newcommand{\kb}[1]{\ket{#1}\bra{#1}}

\newcommand{\PR}{\ensuremath{\mathbb{P}}}

\newcommand{\Tr}{\textnormal{Tr}}
\newcommand{\W}{\mathbf{W}}

\newcommand{\sets}{\mathcal{S}}
\newcommand{\setg}{\mathcal{G}}
\newcommand{\norm}[1]{\left\| #1 \right\|}
\newcommand{\abs}[1]{\left| #1 \right|}
\newcommand{\eps}{\varepsilon}

\title{
On-state commutativity of measurements \\
and joint distributions of their outcomes
}

\author[1]{Jan Czajkowski\thanks{j.czajkowski@uva.nl}}
\author[2]{Alex B. Grilo\thanks{Alex.Bredariol-Grilo@lip6.fr}}
\affil[1]{QuSoft, University of Amsterdam}
\affil[2]{Sorbonne Universit\'{e}, CNRS, LIP6}

\date{}

\begin{document}
\maketitle

\begin{abstract}
In this note, we analyze joint probability distributions that arise from outcomes of sequences of quantum measurements performed on sets of quantum states. 
First, we identify some properties of these distributions that need to be fulfilled to get a classical behavior. Secondly, we prove that a joint distribution exists iff measurement operators ``on-state'' permute (permutability is the commutativity of more than two operators).  
By ``on-state'' we mean properties of operators that hold only on a subset of states in the Hilbert space. 
Then, we disprove a conjecture proposed by 
 Carstens, Ebrahimi, Tabia, and Unruh (eprint 2018), which states that the property of partial on-state permutation implies full on-state permutation.
We disprove this conjecture with a counterexample where pairwise ``on-state'' commutativity does not imply on-state permutability, unlike in the case of commutativity for all states in the Hilbert space.

Finally, we explore the new concept of on-state commutativity by showing a simple proof that if two projections almost on-state commute, then there is a commuting pair of operators that are on-state close to the originals. This result was originally proven by Hastings (Communications in Mathematical Physics, 2019) for general operators.

\end{abstract}
\clearpage
%\tableofcontents
\section{Introduction}

In this work we propose a basic formalism for studying classical distributions that come from joint measurements on quantum states.

Our initial motivation comes from studying a conjecture proposed in a recent paper by Carstens, Ebrahimi, Tabia, and Unruh~\cite{carstens2018quantum}. Their result on quantum indifferentiability\footnote{Indifferentiability is a strong security notion capturing security of cryptographic constructions such as hash functions, where we require that any polynomial-time adversary cannot distinguish if it has access to a cryptographic hash function or an ideal random function, even if she has access to some internal auxiliary functions used to construct the hash function. The quantum version assumes the adversary makes quantum queries.}
%The adversary cannot distinguish the hash function from a random oracle, even though she also has access to the internal compression function used to construct the hash function. The quantum version assumes the adversary makes quantum queries.} 
relies on a conjecture proposed by them, which informally states that commutation of projectors with respect to a fixed quantum state implies a classical joint distribution of their measurement outcomes. More concretely, they conjecture the following.\footnote{See Conjecture~\ref{formal-conjecture-notmod} for the formal statement.}

\begin{conj}[Informal]\label{conj-informal}
If we have a set of $N$ measurements $P_{1},\dots,P_{N}$ that commute on a quantum state  $\ket{\psi},$\footnote{Informally, two operators $A$ and $B$  commute on $\ket{\psi}$ if $[A,B]\ket{\psi} = 0$.} then there exist random variables $X_1,\dots,X_N$ drawn from a distribution $D$ such that for any $ t>1 $, $ \forall _1,\dots, i_t $, the marginals of this distribution on $X_{i_1},\dots,X_{i_t}$ correspond to measuring $\ket{\psi}$ with measurements $P_{i_1},\dots,P_{i_t}$.
\end{conj}

%This conjecture concerns the statistics of outcomes derived from a sequence of quantum measurements, on arbitrary orders. More concretely, https://www.overleaf.com/project/5bc60757e40054586294598f
Motivated by this conjecture,
our goal is to study the behavior of $N$ random variables $X_1,X_2,\dots,X_N$ corresponding to the outcomes of a sequence of quantum measurements that commute on a set of quantum states $\mathcal{F}\subseteq\mathcal{D}(\mathcal{H}) $. Surprisingly, such results have only been studied for $\mathcal{F}=\mathcal{D}(\mathcal{H})$, i.e.\ measurements commuting on {\em all} quantum states.

The focal point of this note is to study which are the necessary and sufficient properties of the quantum setup, so that such a probability distribution is well-defined. With this in hand, we then have two applications. First, we disprove \Cref{conj-informal}. 
Secondly, we show a simpler proof for a variant of the result by Hastings~\cite{hastings2009making} on operators that almost-commute on specific states.
%a version of the main result of \cite{hastings2009making} valid for operators that commute only on-state but limited to projectors, not general Hermitian matrices.

To be able to explain our contributions in more details, we will first start with a detour to very basic properties of probability distributions that arise from classical processes. Then, we discuss how these properties could be defined in the quantum setting (but, unfortunately, they do not hold for general quantum setups), and finally we state our results and discuss related works.

\subsection{Classical Distributions}\label{sec:classical-distr}
We discuss here properties of classical distributions that may be obvious at first but are crucial and not trivial in the quantum world.

%A natural property of $\W_{[N]}$ \agnote{We haven't defined this object, have we?} is the fact that it yields marginal distributions.
%It is natural to expect that summing over some variables we are going to get the probability that the rest of the variables take the given values. As we are going to see later it is not always the case when we discuss distributions of outcomes of quantum measurements.

In the following, we let $A,B,C$ be events that come from a classical experiment. We denote the event corresponding to $A$ not happening as $\overline{A}$, the probability that $A$ and $B$ both happen as $\PR[A,B]$, and the probability that $A$ happens conditioned on the fact that event B happens as $\PR[A|B] = \frac{\PR[A,B]}{\PR[B]}$ (assuming $\PR[B]\neq 0$).

The first property that we want to recall on classical distributions is that we can compute the {\em marginals} of the distribution when given the {\em joint} distribution:
\begin{feat}[Classical Marginals]\label{feat:cmarginals}
     $ \PR[A\mid C]=\PR[A,B \mid C]+\PR[A, \overline{B}\mid C] $.
\end{feat}

A second property that we want to recall is that the probability that $A$ and $\overline{A}$ occur is $0$, even when considering other events:
\begin{feat}[Classical Disjointness]\label{feat:disjointness}
$\PR[A,B\mid\bar{A}]\PR[\overline{A}]=\PR[A, B, \overline{A}] =0$.
\end{feat}

%For example, let us consider the events $ A $ and $ B $. The probability that both events occur is given by  $ \PR[A,B] $. Classically, we can define the marginal probability of $ A $, defined as, where by $ \bar{B} $ is the complement of $ B $ (likewise, we can also define $\PR[B]$).
%We also want to ensure that the marginal distributions exist for all conditional distributions: 
%\agnote{Move the following later?:}The quantum formulation of quantum \emph{marginals} is given in property~\ref{feat:qmarginals}.

%It cannot be true that an event occurs and does not occur. This fact can be written as $\PR[A,B\mid \bar{A}]=\frac{\PR[A\cap B\cap \bar{A}]}{\PR[\bar{A}]}=0$, for events with non-zero weight $ \PR[\bar{A}]\neq 0 $. 
%\agnote{Move later?:}
%We make it formal for quantum distributions with property~\ref{feat:qdisjoint} that we call \emph{disjointness}. We need to state this fact because there might be a sequence in which the probability that events occur such that measuring $A$---followed by some other measurements---and then $\bar{A}$ is not zero.

Another property that we have classically is  \emph{reducibility}, which says that the probability of events $A$ and $A$ both happening is the same as the probability of $A$.

\begin{feat}[Classical Reducibility]\label{feat:reducibility}
$\PR[A,B\mid A]\PR[A]=\PR[A, B, A]=\PR[A,B]$.
\end{feat}
%and define in property~\ref{feat:qreduce}, corresponds to 

Finally, the last property we study is 
%is property~\ref{feat:qseq-ind} that captures the 
\emph{sequential independence}\footnote{Sequential independence has been originally defined in \cite{gudder2001sequential} in the context of quantum measurements.} of random variables. Roughly, this property just says that the probability that event $A$ happens and that event $B$ happens is the same as the probability that event $B$ happens and that event $A$ happens.
%\agnote{Explain the following succinctly:}
%Let us sample some random variables in some order, in a sequence. The question is whether probability of having one sequence is the same as of a similar sequence differing only in order of sampling the variables. As an example let us look at a simple sequence of two events $(A,B)$ (the events might be e.g. evaluating a random variable to some value), it occurs with probability $\PR[A,B]=\PR[B \mid A]\PR[A]$. Note that conditional probability equals $\PR[B \mid A]=\PR[B\cap A]/\PR[A]$. If we considered the probability of another order of events $\PR[B,A]=\PR[A\mid B]\PR[B]$ we would analyze $\PR[A \mid B]=\PR[A\cap B]/\PR[B]$, but $\PR[A\cap B]=\PR[B\cap A]$. Hence there is no difference whether we consider one order of events or another $\PR[A,B]=\PR[B,A]$. Again in the quantum world this would not necessarily hold because conditional probability is not given by the same formula. In our case we require that distribution $\W$ yields the same probability for all sequences of values that differ only by the order of sampling. 

\begin{feat}[Classical Sequential Independence]\label{feat:indepenedence}
$\PR[A\mid B]\PR[B] = \PR[A,B] = \PR[B\mid A]\PR[A]$.
\end{feat}

We stress that these properties hold trivially for {\em all} classical distributions and all events such that $ \PR[A]\neq 0 $, $ \PR[\overline{A}]\neq 0 $, $ \PR[B]\neq 0 $, and $ \PR[C]\neq 0 $.
%, but they do not hold, in generality, quantumly.

\subsection{Quantum Distributions and their Properties}
 %This is the main result of this section, not discussed before.
Our goal is to find necessary conditions for the existence of a classical description of the experiment where we perform 
 a sequence of $ N $ general measurements, 
  irrespective of the order.
 %in an arbitrary order.
 More concretely, we aim to find the properties of measurement operators $Q_1,...,Q_N$ on specific subsets of quantum states $\mathcal{F}$ so that there exists a joint distribution of $X_1,X_2,\dots,X_N$ such that all marginals of this distribution on $X_{i_1},\dots,X_{i_t}$ correspond to measuring a state $\ket{\psi} \in \mathcal{F}$ with measurements $Q_{i_1},\dots,Q_{i_t}$. In this case, we call it a \emph{quantum distribution}.

The main obstacle in this task is the fact that quantum measurements do not necessarily commute, unlike in the classical world: the chosen order for performing the measurements influences the final probability distribution of the joint measurement outcomes.  Because of that, we will consider the quantum analog of Properties~\ref{feat:cmarginals} to~\ref{feat:indepenedence}, and study when such properties hold in the quantum case, and their implication for having such a joint distribution.
Our connections closely follow \cite{muynck1984derivation}, where they show that the existence of a joint distribution for {\em two} arbitrary quantum observables (Hermitian operators) on {\em every} quantum state is equivalent to their commutation. In this work, we show how to extend their analysis in two ways: we are interested in multiple observables and we consider specific sets of quantum states. In order to carry out this analysis, we extend the properties described in \Cref{sec:classical-distr} to quantum measurements and study their relations to each other. We leave the formal definitions of the quantum analogs of these classical properties to \Cref{sec:properties}.

%We are going to highlight some results from the literature that show the requirements for existence of joint distributions.

% First we shortly discuss the current state of affairs and pinpoint the crucial aspects of classical probability distributions that appear to be of relevance to the quantum world. Then we move to the actual subject of this note: a weaker notion of joint distributions \emph{on-state} and their relation to \emph{on-state} commutation. By "on-state" we mean distributions and commutators that are defined with respect to a single quantum state $\psi$. This problem was encountered by Carstens, Ebrahimi, Tabia, and Unruh in their recent paper \cite{carstens2018quantum}.

%In this note we want to describe all requirements that we pose on a function so that it can be considered as a joint probability distribution. We start by identifying some properties of classical joint distributions that we then translate to the quantum world. 

%From a more fundamental perspective our work touches upon the border between the quantum and the classical. We answer the question about existence of joint distributions in the case when we now that probe states are limited to some small subset of the Hilbert space.

\subsection{Our Results}
Using the formalism described in the previous section, we prove the following connections between joint quantum distributions and the measurement operators.

%\agnote{Say a few words for each of them?}
%\begin{result}[Informal statement of \Cref{thm:disjred-iff-projs}]
%If a sequence of quantum measurements on a set of quantum states have the quantum marginal property, then this sequence yields an on-state joint distribution that is disjoint and reducible if and only if all the measurement operators behave like projections.
%\end{result}
%\jcnote{this is more like a lemma, so maybe we can leave it out. In the regular quantum case (for all states) the operators are projectors with some weights, but in the on-state case we only get this 'behave like projectors' property. Basically this is saying that we can assume a bit less then projectors}\agnote{Fine for me.}

First, we show that quantumly, the marginal property also implies the sequential independence one.
\begin{result}[Informal statement of \Cref{thm:marginals-then-seqind}]
	If a joint distribution has the quantum marginal property, then it also has the quantum sequential independence property.
\end{result}

Then, we show that in the on-state case, we have that there is a quantum joint distribution iff all operators permute\footnote{Informally, a set of operators permutes on $\ket{\psi}$ if applied in any order they yield the same state: $A_1\cdots A_N\ket{\psi}=A_{\sigma(1)}\cdots A_{\sigma(N)}\ket{\psi}$, where $\sigma$ is a permutation.}. This result is a generalization of the classic results from \cite{nelson1967dynamical,fine1973probability,fine1982joint,muynck1984derivation} to the on-state case.
\begin{result}[Informal statement of \Cref{thm:joint-perms}]\label{result:2}
Fix a set of quantum states $\mathcal{F}$. 
A set of measurements yield a quantum joint distribution on each state in $\mathcal{F}$ iff these operators permute on every state in $\mathcal{F}$.
\end{result}

Then, we show that pairwise on-state commutation does not imply full on-state permutation, unlike in the case of permutation on all states. This fact---that we prove via a numerical example---together with Result~\ref{result:2} implies that \Cref{conj-informal} is false.
\begin{result}
\Cref{conj-informal} is false.
\end{result}

Finally, our last result is a simpler proof for a restricted version of Theorem~1 in \cite{hastings2009making}, which states that if two operators $A$ and $B$ almost-commute, we can find commuting operators $A'$ and $B'$ that are close to $A$ and $B$, respectively.
 In our case, we consider on-state commutation instead of the regular one, and unlike in \cite{hastings2009making}, our proof works only for projectors.
\begin{result}[Making almost commuting projectors commute]
	Given any two projectors $P_1$ and $P_2$ and a state $\ket{\psi}$ we have that if $\norm{(P_1P_2-P_2P_1)\ket{\psi}}=\epsilon$ then there is a projector $P_2'$ that is close to the original projector on the state $\norm{(P_2'-P_2)\ket{\psi}}\leq\sqrt{2}\epsilon$ and $[P_1,P_2']=0$.
\end{result}

\subsection{Related Work}\label{sec:related-work}
A prominent result in the literature is that a joint distribution for a set of measurements exists iff all the operators pairwise commute.
Different versions of this result were previously proven:
In \cite{nelson1967dynamical} the author considers the case of continuous variables and $N$ observables. A similar result but without specifying the Hilbert space is achieved with different mathematical tools in \cite{fine1982joint}. In the specific case where we have only two observables, we mention three works; In \cite{fine1973probability} and \cite{muynck1984derivation} the authors prove the classic problem in a way similar to each other, but using different mathematical tools. All but the first work mentioned here focus on the joint distribution as a functional from the space of states. An approach using $*$-algebras was presented by Hans Maassen in
~\cite{HassenNotes,maassen2010quantum}.

The authors of \cite{gudder2002sequentially} analyze the case of general measurements but prove that the measurement operators pairwise commute iff the square-root operators permute (Corollaries 3 and 6 in \cite{gudder2002sequentially}), in the sense of our Definition~\ref{def:onstate-perms} (for all states in $ \mathcal{H} $). In general the problem of conditional probabilities in Quantum Mechanics was discussed by Cassinelli and Zanghi in \cite{cassinelli1983conditional}.

The related problems of incompatible devices measurement and joint measurability of quantum effects are covered in \cite{heinosaari2016invitation} and \cite{bluhm2018joint} respectively.

In \cite{lin1997almost, friis1996almost} the authors prove that for any two Hermitian matrices if their commutator has small norm, then there are operators close to the originals that fully commute. In \cite{hastings2009making} Hastings proves how close the new operators are in terms of the norm of the commutator.

\subsection*{Organization}
In~\Cref{sec:prelim}, we provide some preliminaries. Then in \Cref{sec:main-distributions}, we discuss quantum distributions and their properties. Finally, in \Cref{sec:almost-commuting}, we discuss the almost-commuting case.

\subsection*{Acknowledgements}
JC thanks Dominique Unruh and Christian Schaffner for helpful discussions.
JC was supported by a NWO VIDI grant (Project No. 639.022.519).
Most of this work was done while AG was affiliated to CWI and QuSoft.

\section{Preliminaries}\label{sec:prelim}

\subsection{Notation}
In this work, we are going to use calligraphic letters ($\mathcal{S},\mathcal{R},...$) to denote sets. We denote $[N]:=\{1,2,\dots,N\}$. 
For $\mathcal{S}\subseteq[N]$,  we denote by $\mathcal{S}(i)$ the $i$-th element of the set $\mathcal{S}$ in ascending order. For some fixed sets $\mathcal{X}_1,...,\mathcal{X}_N$,
we denote by $\vec{x}$ an element of $\mathcal{X}_1\times\cdots\mathcal{X}_N$ and for $\mathcal{S}\subseteq [N]$ we have $\vec{x}_{\mathcal{S}}:=(x_{\mathcal{S}(1)},\dots,x_{\mathcal{S}(|\mathcal{S}|)})$.
We denote the set of all $t$-element permutations by $\Sigma_t$. For some complex number $c = a + b\textnormal{i}$, we define $\mathfrak{Re}(c) = a$ as its real part.

\subsection{Quantum Measurements\label{sec:quantum}}
We briefly review some concepts in quantum computation/information and we refer to \cite{nielsen2002quantum}
for an more detailed introduction to these topics.

Quantum states are represented by positive semi-definite operators with unit trace, i.e.,  $\rho\succeq 0, \Tr(\rho)=1$. We denote the set of all density operators on some Hilbert space $\mathcal{H}$ by $\mathcal{D}(\mathcal{H})$.

To describe general measurements, we use the notion of Positive Operator Valued Measure (POVM). The only requirement of POVMs is that they consist of positive operators and sum up to the identity operator.
More formally, a POVM with set of outcomes $\mathcal{X}$ is described by a set of operators $\mathcal{M}=\{ Q^x \}_{x\in\mathcal{X}}$, where $\forall x\in\mathcal{X}:Q^x\succeq 0$, and $\sum_{x\in\mathcal{X}} Q^x=\mathbbm{1}$.

We denote the probability of getting the outcome $x$ when measuring $\rho$ with the measurement $\mathcal{M}$ by $\PR[x\gets \mathcal{M}(\rho)]:= \Tr (Q^{x}\rho)$. To describe the post-measurement state, we can write down operators of $\mathcal{M}$ as products of linear operators on $\mathcal{H}$ (denoted by $\mathcal{L}(\mathcal{H})$), $ Q^x=A^{x\dagger}A^x$, where $A^x\in\mathcal{L}(\mathcal{H})$ (such a decomposition is always possible since $Q^x \succeq 0$). The post-measurement state when the outcome of $\mathcal{M}$ on $\rho$ is $x$ is given by
\begin{equation}
\rho_x:= \frac{A^x\rho A^{x\dagger}}{\Tr (Q^x\rho)}.
\end{equation}
The operator $A^x$ is called the \emph{square root} of $Q^x$.

\section{Quantum Distributions}
\label{sec:main-distributions}
In this section, we study the description of the statistics of outcomes derived from a sequence of measurements. 
Our approach is to consider the quantum version of the classical properties described in \Cref{sec:classical-distr}. Since quantum measurements do not commute in general, these quantum properties do not always hold. We then study the connection between properties of the measurements and the properties of their outcome distribution.

%Our approach is to consider several properties of joint distributions that although obvious classically, but not quantumly. We back our definitions by examples on classical events. Then we connect these properties---of distributions---to properties of measurement operators. 

%As we discussed, the main obstacle here is the fact that quantum measurements do not commute: unlike in the classical world, the chosen order for performing the measurements influences the final probability distribution of the joint measurement outcomes. Our solution is then to put constraints on the total joint distribution coming from the classical world and force the connection to the quantum world for the particular measurement operators, that we know how to define and assume existence of appropriate measurements for the sequences longer than one. These constraints are the properties of joint distributions, already listed in section~\ref{sec:classical-distr}.

The structure of our proofs follows \cite{muynck1984derivation}, where they show that for two Hermitian observables there is a joint distribution for the outcomes of their joint measurement iff they commute. We stress that the result in \cite{muynck1984derivation} only works for measurements that commute on every quantum state and our result extends it to the case of joint distributions defined on a limited set of states.

In the following, we denote the observables with $Q$ and their square-roots with $R$.
In Section~\ref{sec:properties}, we define the quantum analogues of the classical properties of distributions defined in \Cref{sec:classical-distr}. Then, in \Cref{sec:distr-perm}, we state and prove the main result of this section, where we show a connection between existence of a distribution and permutability---a generalization of commutativity---of the corresponding measurement operators.

\subsection{Quantum Distributions}\label{sec:properties}
We analyse a functional from the set of density operators and finite sets of $N$ random variables $ X_1,X_2,\dots,X_N $ to reals $\W_{[N]}: \mathcal{D}(\mathcal{H})\times (\mathcal{X}_1\times\cdots\times \mathcal{X}_N)\to [0,1]$. We define this functional\footnote{Note that the superscript of $\W^{\rho}_{[N]}(\vec{x})$ denotes the first input to the functional, so we have $\W_{[N]}(\rho,\vec{x})$.} as
\begin{equation}
\W_{[N]}^{\rho}(\vec{x}) := \Tr \left( Q_{[N]}^{\vec{x}} \rho\right),\label{eq:wfuncdef}
\end{equation}
where $Q_{[N]}^{\vec{x}}$ is a positive semidefinite operator corresponding to the outcome $\vec{x}\in\mathcal{X}_1\times\cdots\times \mathcal{X}_N$. The subscript $[N]$ of $\W$ denotes the set of indices of random variables that the distribution is defined on. This definition is similar to the one proposed by \cite{muynck1984derivation}.

The starting point of our discussion is that every random variable $ X_i $ corresponds to the measurement $ \mathcal{M}_i=\{Q_i^{x_i}\} $. So we have to keep in mind that these operators are fixed throughout this note.

Given the definition of $\W$, we can state conditions so that it can be seen as a joint quantum distribution:
\begin{itemize}
% 	\item \agnote{Where are these random variables used? Shouldn't they appear somewhere in the equation, otherwise we are just repeating the definition?} Random variables $ X_1,X_2,\dots,X_N $ are described by measurement  $\{Q^{\vec{x}}_{[N]}\}_{ \vec{x}\in\mathcal{X}_1\times\cdots\times\mathcal{X}_N}$ such that 
% 	$\forall\rho\in\mathcal{D}(\mathcal{H})$, we have
% 	\[
% 	%\forall  \exists Q^{\vec{x}}_{[N]}\preceq 0 
% 	 \W^{\rho}_{[N]}(\vec{x})=\PR[\vec{x}\gets Q^{\vec{x}}_{[N]}(\rho)]=\Tr \left(Q^{\vec{x}}_{[N]}\rho\right). \]
	\item Normalization: 
	$ \sum_{\vec{x}}Q^{\vec{x}}_{[N]}=\mathbbm{1} $, which implies that
	for all $\rho\in\mathcal{D}(\mathcal{H}), \; \sum_{\vec{x}}\W^{\rho}(\vec{x})=1$.
	\item Linearity: for every $\vec{x}\in\mathcal{X}_1\times\cdots\times\mathcal{X}_N$,  $\rho_1,\rho_2\in\mathcal{D}(\mathcal{H})$ and  $\lambda_1,\lambda_2\in [0,1]$, we have that \[\W^{\lambda_1\rho_1+\lambda_2\rho_2}_{[N]}(\vec{x})=\lambda_1\W^{\rho_1}_{[N]}(\vec{x})+\lambda_2\W^{\rho_2}_{[N]}(\vec{x}).\]
	\item Non-negativity:  for every $\vec{x}\in\mathcal{X}_1\times\cdots\times\mathcal{X}_N$ and $\rho\in\mathcal{D}(\mathcal{H})$, we have $\W^{\rho}_{[N]}(\vec{x})\geq 0$.
\end{itemize}

We describe the quantum analogues of the properties described in \Cref{sec:classical-distr}.

\paragraph{Marginals.}
We start with a sequence of general measurements (i.e.\ POVMs): $ \{\mathcal{M}_i\}_{i\in [N]} $, where $ \mathcal{M}_i:=\{Q_i^{x}\}_{x\in\mathcal{X}_i} $ for all $ i $ and all $ Q_i^{x}\succeq 0 $ and $ \sum_{x\in\mathcal{X}_i}Q^{x}_i=\mathbbm{1} $. Moreover,  for $ \sets\subseteq[N] $ we define $ Q^{\vec{y}}_{\sets} $ to be measurement operators  where $ \vec{y}\in\mathcal{X}_{\sets(1)}\times\cdots\times\mathcal{X}_{\sets(|\sets|)} $.

Given $Q^{\vec{y}}_{\sets}$ and their corresponding square-roots $ R^{\vec{y}}_{\sets} $ and $\rho$,
we have that if $ \Tr \left(R^{\vec{y}}_{\mathcal{S}}\rho R^{\vec{y}\dagger}_{\mathcal{S}} \right)\neq 0 $, then we define the conditional distribution for any sequence $\vec{x}$ as
\begin{equation}\label{key}
\W_{[N]}^{\rho}(\vec{x}\mid \vec{y}):= \W_{[N]}^{R^{\vec{y}}_{\mathcal{S}}\rho R^{\vec{y}\dagger}_{\mathcal{S}}}(\vec{x})/\Tr R^{\vec{y}}_{\mathcal{S}}\rho R^{\vec{y}\dagger}_{\mathcal{S}}.
\end{equation}

For all those measurements, for a set $ \mathcal{F}\subseteq\mathcal{D}(\mathcal{H}) $, and for $\mathcal{U}\subseteq[N]$, we define the ``orbit'' of the post-measurement states.
For any $\mathcal{T}\subseteq[N]$ of size $t$, we take  $s\leq t$ sets $\mathcal{S}_1,...,\mathcal{S}_s$ that are a  partition of $\mathcal{T}$. We then consider the post-measurement states generated by sequences of measurements corresponding to $\mathcal{S}_i$: 
    \begin{align}\label{eq:setg-def}
	&\setg_{\mathcal{U}}(\mathcal{F}) :=  \left\{ R^{\vec{y}_s}_{\sets_s}\cdots R^{\vec{y}_1}_{\sets_1}\psi R^{\vec{y}_1\dagger}_{\sets_1}\cdots R^{\vec{y}_s\dagger}_{\sets_s}/\Tr \left(R^{\vec{y}_s}_{\sets_s}\cdots R^{\vec{y}_1}_{\sets_1}\psi R^{\vec{y}_1\dagger}_{\sets_1}\cdots R^{\vec{y}_s\dagger}_{\sets_s}\right):\psi\in\mathcal{F}, \mathcal{T}\subseteq\mathcal{U},\right.\nonumber\\
	&\left. s\leq \abs{\mathcal{T}}, \;  \sets_1,...,\sets_s\subseteq\mathcal{T} , \bigcup_{i=1}^s\sets_i=\mathcal{T}, \forall i\neq j \, \mathcal{S}_i\cap\mathcal{S}_j=\emptyset, \vec{y}_i\in\mathcal{X}_{\sets_i(1)}\times\cdots\times\mathcal{X}_{\sets_i(|\sets_i|)} \right\},
\end{align}
where $ R^{\vec{y}_i}_{\sets_i} $ are the square-root operators of $ Q^{\vec{y}_i}_{\sets_i}=R^{\vec{y}_i\dagger}_{\sets_i}R^{\vec{y}_i}_{\sets_i} $. The subscript of $\setg$ denotes the set we take the subsets of, usually it is $[N]$ but later we also consider $[N]\setminus\sets$ for some $\sets$.

With our quantum marginals property, we require that the operator we get after we sum over a subset of variables is still a valid measurement operator.
\begin{feat}[Quantum Marginals]\label{feat:qmarginals}
We say that the joint distribution $\W$ has the \emph{quantum marginals property on set $\mathcal{F}$} if 
for every $\mathcal{S} \subseteq [N]$,
there is a measurement $ \mathcal{M}_{\sets} = \{Q^{\vec{y}}\}_{\vec{y} \in (\mathcal{X}_i)_{i \in \mathcal{S}}}$
such that for every value $ \vec{x}\in \mathcal{X}_1\times \mathcal{X}_2\times\cdots\times \mathcal{X}_N  $, denoting $ \vec{x}:=(x_1,x_2,\dots,x_N) $
%and $ \vec{y}:= (x_{\mathcal{S}(1)},x_{\mathcal{S}(2)},\dots,x_{\mathcal{S}(t)}) $
 and for every density operator $ \rho\in\setg_{[N]}(\mathcal{F})$ defined as in Equation~\eqref{eq:setg-def} we have that 
	\begin{align}
	&\W^{\rho}_{\mathcal{S}}(\vec{x}_{\mathcal{S}}):= \Tr \left(Q^{\vec{x}_{\mathcal{S}}}_{\mathcal{S}}\rho\right)=\sum_{x_i\in \mathcal{X}_i,i\in[N]\setminus \mathcal{S} }\W^{\rho}_{[N]}(\vec{x}).
	\end{align}
	Additionally for $|\mathcal{S}| = 1$ the operators $ Q^{x_i}_{i} $ are the operators from $ \mathcal{M}_{i} $.
\end{feat}

\paragraph{Disjointness.}

It follows from the definition of POVMs that the quantum measurement operators need not be orthogonal, and this implies that the disjointness property (Property~\ref{feat:disjointness} does not hold in generality quantumly.

Disjointness is a property that concerns a post-measurement state of a set $\mathcal{S}$ of variables. To ensure the existence of a measurement operator corresponding to $\mathcal{S}$, we need to assume Property~\ref{feat:qmarginals}.
\begin{feat}[Quantum Disjointness]\label{feat:qdisjoint}
Let $\W$ be a
 joint distribution for which Property~\ref{feat:qmarginals} holds. We say that $\W$ has the \emph{quantum disjointness property on set $\mathcal{F}$} if for every subset $\mathcal{S} \subseteq [N]$, for every density operator  $ \rho\in\setg_{[N]\setminus\sets}(\mathcal{F}) $, and for every value $ \vec{x}\in \mathcal{X}_1\times \mathcal{X}_2\times\cdots\times \mathcal{X}_N  $ and $\vec{y} \in \prod_{i \in \mathcal{S}}\mathcal{X}_i$, we have that
if $ \vec{y}\neq \vec{x}_{\mathcal{S}}$, then 
\begin{align}
		&\W^{\rho}_{[N]}(\vec{x}\mid \vec{y})\W^{\rho}_{[N]}( \vec{y})=\Tr \left(Q^{\vec{x}}_{[N]} R^{\vec{y}}_{\mathcal{S}} \rho  R^{\vec{y}\dagger}_{\mathcal{S}}\right)=0.
	\end{align}
\end{feat}

\paragraph{Reducibility.}
Reducibility (Property~\ref{feat:reducibility}) is a similar property to disjointness but with the key difference that we condition on the same event:
	\begin{feat}[Quantum Reducibility]\label{feat:qreduce}
Let $\W$ be a joint distribution for which Property~\ref{feat:qmarginals} holds. We say that $\W$ has the \emph{quantum reducibility property on set $\mathcal{F}$}
if for every subset $\mathcal{S} \subset [N]$, every density operator $ \rho\in\setg_{[N]\setminus\sets}(\mathcal{F}) $, 
and value $ \vec{x}\in \mathcal{X}_1\times \mathcal{X}_2\times\cdots\times \mathcal{X}_N  $, we have that 
		\begin{align}
			&\W^{\rho}_{[N]}(\vec{x}\mid\vec{x}_{\mathcal{S}})\W^{\rho}_{[N]}(\vec{x}_{\mathcal{S}})=\Tr \left( Q^{\vec{x}}_{[N]}R^{\vec{x}_{\mathcal{S}}}_{\mathcal{S}}
			\rho
			R^{\vec{x}_{\mathcal{S}}\dagger}_{\mathcal{S}}\right)=\Tr\left( Q^{\vec{x}}_{[N]}
			\rho
			\right).
		\end{align}
	\end{feat}

\phantom{new line}

Note that the last two properties together allow us to conclude that the operators are (morally) \emph{on-state projections}: Property~\ref{feat:qdisjoint} plays the role of different projectors being orthogonal and Property~\ref{feat:qreduce} that projecting twice to the same space does not change the resulting state.

More concretely, we say that the $R_i$'s are on-state projectors on $\psi\in\mathcal{F}$ if for all $\mathcal{S}\subseteq[N]$, all $\vec{x}\in \mathcal{X}_1\times\cdots\times \mathcal{X}_N$, all $\vec{y} \in \prod_{i \in \mathcal{S}}\mathcal{X}_i$, and for $R_{\mathcal{S}}^{\vec{y}}:=R^{y_{1}}_{\mathcal{S}(1)}R^{y_{2}}_{\mathcal{S}(2)} \cdots R^{y_{t}}_{\mathcal{S}(t)} $ and $Q_{\mathcal{S}}^{\vec{y}}:=R_{\mathcal{S}}^{\vec{y}\dagger}R_{\mathcal{S}}^{\vec{y}} $ (similarly for $[N]$), we have 
\begin{align}\label{eq:on-state-projectors}
	&\Tr \left(R^{\vec{y}\dagger}_{\sets} Q^{\vec{x}}_{[N]}R^{\vec{y}}_{\sets} \psi\right) = \delta_{\vec{y},\vec{x}_{\sets}} \Tr \left(Q^{\vec{x}}_{[N]}\psi\right).
\end{align}

\paragraph{Sequential Independence}
As previously discussed, the notion of time order in the quantum setting  is much more delicate as the probabilistic events no longer commute. Let us go back to the example of the simple sequence $(A,B)$ from Section~\ref{sec:classical-distr} but now consider $A$ and $B$ as quantum observables measured on the state $\rho$. Let us assume for simplicity that $A$ and $B$ are projections. The probability of measuring $a$ with $A$ is $\Tr \left(A \rho\right)$ and the state after this measurement is $\rho_a:=\frac{A\rho A}{\Tr \left(A\rho\right)}$ so the probability of measuring the sequence $(a,b)$ equals
\begin{equation}
\PR[b\gets B(\rho_a)]\PR[a\gets A(\rho)]=\Tr \left(B \frac{A\rho A}{\Tr (A\rho)} \right) \Tr \left(A\rho\right)=\Tr \left(ABA\rho\right).\label{eq:abarho}
\end{equation}
On the other hand the probability of measuring the sequence $(b,a)$ equals $\Tr \left( BAB\rho \right)$ which is in general different than Equation~\eqref{eq:abarho}. This simple example shows that sequential independence is not attained by all quantum joint probabilities. More formally, the notion of sequential independence from \cite{gudder2002sequentially} for quantum joint probabilities can be stated as follows.

\begin{feat}[Quantum Sequential Independence]\label{feat:qseq-ind}
Let $\W$ be a
 joint distribution for which Property~\ref{feat:qmarginals} holds. We say that $\W$ has  the \emph{quantum sequential independence property on set $\mathcal{F}$}
if 
for every density operator $ \psi\in\mathcal{F} $, for any $\mathcal{T}\subseteq[N]$ of size $t$, for all $s\leq t$, partition $\sets_1,...,\sets_s$ of $\mathcal{T}$, permutation $ \sigma \in \Sigma_{s}$, and $ \vec{x}\in \mathcal{X}_1\times \mathcal{X}_2\times\cdots\times \mathcal{X}_N  $ such that $ \vec{x}:=(x_1,x_2,\dots,x_N) $ and $ \vec{y}_i:= \left(x_{\mathcal{S}_i(1)},x_{\mathcal{S}_i(2)},\dots,x_{\mathcal{S}_i(\abs{\sets_i})}\right) $,  we have that 
	\begin{align}
		\Tr &\left(Q^{\vec{y}_s}_{\sets_s}  R^{\vec{y}_{s-1}}_{\sets_{s-1}} \cdots R^{\vec{y}_{1}}_{\sets_{1}}\psi R^{\vec{y}_{1}\dagger}_{\sets_{1}}R^{\vec{y}_{2}\dagger}_{\sets_{2}}\cdots R^{\vec{y}_{s-1}\dagger}_{\sets_{s-1}} \right)\nonumber\\
		& =	\Tr \left(Q^{\vec{y}_{\sigma(s)}}_{\sets_{\sigma(s)}}  R^{\vec{y}_{\sigma(s-1)}}_{\sets_{\sigma(s-1)}} \cdots R^{\vec{y}_{\sigma(1)}}_{\sets_{\sigma(1)}}\psi R^{\vec{y}_{\sigma(1)}\dagger}_{\sets_{\sigma(1)}}R^{\vec{y}_{\sigma(2)}\dagger}_{\sets_{\sigma(2)}}\cdots R^{\vec{y}_{\sigma(s-1)}\dagger}_{\sets_{\sigma(s-1)}}\right).
	\end{align}
\end{feat}

\begin{thm}[Marginals imply Sequential Independence]\label{thm:marginals-then-seqind}
	If a joint distribution $\W$ has Property~\ref{feat:qmarginals}, \ref{feat:qdisjoint}, and \ref{feat:qreduce}, then it also has Property~\ref{feat:qseq-ind}.
\end{thm}
\begin{proof}
	First note that Properties~\ref{feat:qmarginals}, \ref{feat:qdisjoint}, and \ref{feat:qreduce} directly imply Equation~\eqref{eq:on-state-projectors}.
	
	We now prove the statement for two sets of indices and then argue how to extend it for $s$ sets. In our restricted case, we have that
	\begin{align}
		& \Tr \left(Q^{\vec{y}_1}_{\sets_1} R^{\vec{y}_2}_{\sets_2}\psi R^{\vec{y}_2\dagger}_{\sets_2}\right) =\Tr \left(\sum_{\vec{y}'}  Q^{\vec{y}'}_{\sets_1\cup\sets_2} R^{\vec{y}_2}_{\sets_2}\psi R^{\vec{y}_2\dagger}_{\sets_2}\right)\\
		& =\Tr  \left(Q^{\vec{y}}_{\sets_1\cup\sets_2} R^{\vec{y}_2}_{\sets_2}\psi R^{\vec{y}_2\dagger}_{\sets_2} \right) \\
		& =\Tr  \left(Q^{\vec{y}}_{\sets_1\cup\sets_2} \psi\right),\label{eq:marg-seq-step}
	\end{align}
where the	first equality comes from the marginals property,
the second equality comes from the disjointness property if we take $ \vec{y} $ that agrees with $ \vec{y}_2 $ on $ \sets_{2} $, and the last equality comes from the reducibility property. 

The above derivation can be repeated for any other two sets $\sets'_1$ and $\sets'_2$ such that $\sets'_1\cup\sets'_2=\sets_1\cup\sets_2$. Any pair like this will yield $\Tr  \left(Q^{\vec{y}}_{\sets_1\cup\sets_2} \psi\right)$ that equals  $\Tr  \left(Q^{\vec{y}}_{\sets'_1\cup\sets'_2} \psi\right)$ for all other sets $\sets'_1$ and $\sets'_2$, hence we have proven Property~\ref{feat:qseq-ind} for any two subsets.

In general, we prove a slightly stronger statement than Property~\ref{feat:qseq-ind}. We prove that not only different sequences give the same probabilities but also that these probabilities equal $\Tr  \left(Q^{\vec{y}}_{\mathcal{T}} \psi\right)$ for some operator $Q^{\vec{y}}_{\mathcal{T}}$.

The general case holds by taking $ \rho\in\setg_{[N]\setminus\sets_s}(\mathcal{F}) $ instead of $\psi$ and use it in the above calculation. To prove Property~\ref{feat:qseq-ind} for any $s$ sets we consider $\rho=  R^{\vec{y}_{s-1}}_{\sets_{s-1}} \cdots R^{\vec{y}_{1}}_{\sets_{1}}\psi R^{\vec{y}_{1}\dagger}_{\sets_{1}}R^{\vec{y}_{2}\dagger}_{\sets_{2}}\cdots R^{\vec{y}_{s-1}\dagger}_{\sets_{s-1}}$ and $\Tr\left(Q^{\vec{y}_s}_{\sets_s}\rho\right)$. We shave off operators from $\rho$ one by one using Equation~\eqref{eq:marg-seq-step}. After $s$ steps we have that $\Tr\left(Q^{\vec{y}_s}_{\sets_s}\rho\right)= \Tr\left(Q^{\vec{x}_{\mathcal{T}}}_{\mathcal{T}}\psi\right)$. Again, repeating this procedure for a different $\rho'$ with sets that sum to the same $\mathcal{T}$ yields the claimed result.
\end{proof}

\subsection{Joint Distributions and Permutability}\label{sec:distr-perm}

In this section, we state the definition of a joint distribution that describes a sequence of quantum measurements done on states from a small set. We also define a generalization of commutativity and prove that a joint distribution exists if and only if the measurement operators are permutable.

\begin{defi}[Quantum Joint Distribution On State]\label{def:quantdist}
A joint distribution of  $ N $ random variables $ X_1,\dots, X_N $ that describe outcomes of general quantum measurements of states $ \psi\in\mathcal{F} $ is defined as a positive, normalized, and linear functional 
\begin{align}
	&\W_{[N]}: \mathcal{D}(\mathcal{H})\times(\mathcal{X}_{1}\times\mathcal{X}_{2}\times\cdots\times\mathcal{X}_{N})\to[0,1],
\end{align}
for which
\begin{enumerate}[label={(\arabic*)}]
\item Quantum Marginals on states in $\mathcal{F}$, Property~\ref{feat:qmarginals}, holds,
\item Quantum Disjointness on states in $\mathcal{F}$, Property~\ref{feat:qdisjoint}, holds,
\item Quantum Reducibility on states in $\mathcal{F}$, Property~\ref{feat:qreduce}, holds,
\item Quantum Sequential Independence on states in $\mathcal{F}$, Property~\ref{feat:qseq-ind}, holds.
\end{enumerate}
\end{defi}

Next we show the connection between the existence of joint distributions and requirements on the measurement operators.
But first let us define the notion of on-state permutability.
\begin{defi}[On-state permutator, (fully) permutable operators]\label{def:onstate-perms}
	For any $s$ operators $R_i$ and a permutation of the $s$-element set $\sigma\in\Sigma_s$ the permutator on $\psi$ is defined as 
	\begin{align}
	[ R_1,& R_2, \dots, R_s ]_{\psi}(\sigma) := \nonumber \\
	&\Tr\left( R_{s}R_{s-1}\cdots R_{1}\psi R^{\dagger}_{1}R^{\dagger}_{2}\cdots R^{\dagger}_{s}\right) - \Tr \left( R_{\sigma(s)}R_{\sigma(s-1)}\cdots R_{\sigma(1)}\psi R^{\dagger}_{\sigma(1)}R^{\dagger}_{\sigma(2)}\cdots R^{\dagger}_{\sigma(s)}\right).
	\end{align}
We say that the operators  $R_1, \dots, R_s$ are \emph{permutable} if $[ R_1, \dots, R_s ]_{\psi}(\sigma) =0$ for all $\sigma\in\Sigma_s$. Moreover, we call a set of measurements $\{\mathcal{M}_i\}_i$ \emph{fully permutable on $\mathcal{F}$} if all square-root operators $ R^{\vec{y}}_i $ of these measurements $ \mathcal{M}_i $ permute on $ \psi\in\mathcal{F} $ for all $\sigma\in\Sigma_s$.
\end{defi}

Now we state the theorem connecting existence of joint distributions with permutability of the measurement operators. This statement extends Theorem $3.2$ of \cite{muynck1984derivation}, where they prove that if the joint distribution satisfies the marginals property (they use the name ``nondisturbance''), then the operators pairwise commute. 
\begin{thm}[Quantum Joint Distribution and Permutability]\label{thm:joint-perms}
	There is a quantum joint distribution on states $ \psi\in\mathcal{F} $ describing measurement outcomes of $N$ observables $X_1,\dots,X_N$ if and only if all square roots $ R^{x}_i $ of $ Q^{x}_i $ operators of measurements $ \mathcal{M}_i $ permute on $ \psi\in\mathcal{F} $ and are on-state projectors according to Equation~\eqref{eq:on-state-projectors}.
\end{thm}
\begin{proof}
	
	$(\Rightarrow)$
	Permutability follows from Property~\ref{feat:qseq-ind} for $ N $ single-element sets $ \sets_i=\{i\} $. Being on-state projectors (Equation~\eqref{eq:on-state-projectors}) follows from Property~\ref{feat:qdisjoint} and Property~\ref{feat:qreduce}. 
	
	$(\Leftarrow)$
	The other direction of the proof follows by setting the measurement operators to $Q_{\mathcal{S}}^{\vec{y}}:=R^{y_t\dagger}_{\mathcal{S}(t)}R^{y_{t-1}\dagger}_{\mathcal{S}(t-1)}\cdots
	 R^{y_{1}\dagger}_{\mathcal{S}(1)}R^{y_{1}}_{\mathcal{S}(1)}R^{y_{2}}_{\mathcal{S}(2)} \cdots R^{y_{t}}_{\mathcal{S}(t)} $, for every set $ \sets \subseteq [N] $ with $ |\sets|=t $. Similarly we define 
	 $R_{\mathcal{S}}^{\vec{y}}:=R^{y_{1}}_{\mathcal{S}(1)}R^{y_{2}}_{\mathcal{S}(2)} \cdots R^{y_{t}}_{\mathcal{S}(t)} $.
	
	The marginals property follows from the fact that $\sum_{x_{i}\in\mathcal{X}_i}Q_i^{x_i}=\mathbbm{1}$: 
	\begin{align}
		&\sum_{x_i\in \mathcal{X}_i,i\in[N]\setminus \mathcal{S} }\Tr \left(Q^{x_{i_i}}_{i_1} R^{\vec{y}}_{[N]\setminus\{i_1\}} \rho R^{\vec{y}\dagger}_{[N]\setminus\{i_1\}} \right) \nonumber\\
		&= \sum_{x_i\in \mathcal{X}_i,i\in[N]\setminus (\mathcal{S}\cup\{i_1\}) }\Tr \left(Q^{x_{i_2}}_{i_2} R^{\vec{y}}_{[N]\setminus\{i_1,i_2\}} \rho R^{\vec{y}\dagger}_{[N]\setminus\{i_1,i_2\}} \right) = \cdots = \Tr \left(Q^{\vec{y}}_{\mathcal{S}}\rho\right).
	\end{align}
	
	Properties~\ref{feat:qdisjoint} and~\ref{feat:qreduce} are natural consequences of Equation~\eqref{eq:on-state-projectors}. The only difference between the properties and on-state projections is the set $\setg_{[N]\setminus\sets}(\mathcal{F})$ versus just $\mathcal{F}$. Nonetheless, with our definition of $R_{\mathcal{S}}^{\vec{y}}$, Equation~\eqref{eq:on-state-projectors} implies disjointness and reducibility.
	
	In Theorem~\ref{thm:marginals-then-seqind} we have already proved that marginals together with disjointness and reducibility imply sequential independence, which concludes our proof.
\end{proof}

\subsection{Pairwise on-state commutation does not imply full commutation}\label{sec:counterexample}
We now investigate whether full permutability is the weakest assumption we can have for joint distributions to exist.

When we consider this question for the full Hilbert space $ \mathcal{F}=\mathcal{D}(\mathcal{H}) $, this problem has been considered by a number of works in the literature~\cite{nelson1967dynamical,fine1973probability,fine1982joint,muynck1984derivation} and it is well-known that it suffices for the measurement operators to pairwise commute, i.e., pairwise commutation on all possible quantum states implies permutability of the operators.

Our goal in this section is to consider the case where $ \mathcal{F} \subsetneq \mathcal{D}(\mathcal{H}) $. In particular, in \cite{carstens2018quantum}, in order to connect perfect quantum indifferentiability to classical indifferentiability with stateless simulators\footnote{
Roughly, we say that $A$ is classical (quantum) indifferentiable from $B$ iff we can map classical (resp. quantum) attacks on $A$ to classical (resp. quantum) attacks on $B$ using simulators. Moreover, we say that the simulator is stateless if it does not store any internal state.
%Indifferentiability is an important property in cryptography since it allows to carry on security properties from $B$ to $A$.
%In the proof, they analyze the quantum simulator with respect to her outputs, that depend on the outcomes of  $q-1$ projective measurements on the simulator's state. Finally, they construct a classical stateless algorithm that is tasked with sampling from the (conjectured) joint distribution on $N$ outcomes.
}, they rely on the following conjecture.

\begin{conj}[Conjecture 2 from \cite{carstens2018quantum}]\label{formal-conjecture-notmod}
Consider $N$ binary measurements described by projectors $P_1,\dots,P_N$, and a quantum state $\ket{\Psi}$.

Assume that any $t$ out of the $N$ measurements permute on state $\ket{\Psi}$. That is, for any $I$ with $|I|=t$, if $P'_1,\dots,P'_t$ and $P^{\prime \prime}_1,\dots,P^{\prime \prime}_t$ are the projectors $\{P_i\}_{i\in I}$ (possibly in different order), then $P^{\prime}_t,\dots,P'_1\ket{\Psi}=P^{\prime \prime }_t,\dots,P^{\prime\prime }_1\ket{\Psi}$.

Then there exist random variables $X_1,\dots,X_N$ with a joint distribution $D$ such that for any $I=\{i_1,\dots,i_t\}$ the joint distribution of $X_{i_1},\dots,X_{i_t}$ is the distribution of the outcomes when we measure $\ket{\Psi}$ with measurements $P_{i_1},\dots,P_{i_t}$.
\end{conj}

Conjecture~\ref{formal-conjecture-notmod} states that if any $t$ measurement operators permute on state $\ket{\Psi}$ then there is a joint distribution. From Theorem~\ref{thm:joint-perms}, we know that if there is a joint distribution then the operators fully permute.  Hence, the key point of the conjecture is that if any $t$ operators permute on a state, then they fully permute on it.
However, we show here that \Cref{formal-conjecture-notmod} is not true, in general. 
\begin{thm}\label{lem:counterexample}
    There is a set of four projectors $\{P_1,P_2,P_3,P_4\}$ and a state $\ket{\phi}\in\mathbb{C}^8$ such that the projectors are 2-permutable (they pairwise commute) on state $\ket{\phi}$ and they are \emph{not} 4-permutable on $\ket{\phi}$.
\end{thm}

\begin{proof}
To prove the statement we found an example of such operators and a state numerically by random constrained search.
We consider 4 projectors $P_i$ and a state $\ket{\phi}$ of dimension 8. The constraints we impose are 
\begin{align}
    \forall i,j\neq i \; [P_i,P_j]\ket{\phi}=0,
\end{align}
moreover operators $P_i$ are projectors: $\forall i P_i^2=P_i$ and $\ket{\phi}$ is a unit-norm complex vector.

We look for an example to the statement that 2-permutability (commutativity) does not imply 4-permutability, so that
\begin{align}
    (P_1P_2P_3P_4-P_3P_4P_1P_2)\ket{\phi}\neq 0.
\end{align}
To find such example we used software for symbolic computing to define the problem and maximize $\|(P_1P_2P_3P_4-P_3P_4P_1P_2)\ket{\phi}\|$, where we maximize over operators $P_i$ and the state $\ket{\phi}$.

The result of our optimization can be found in Appendix~\ref{sec:appendix-num}. Note that however we consider vector equalities---instead of just traces like in Theorem~\ref{thm:joint-perms}---our example provides a stronger argument for the necessity of full permutability.
\end{proof}

One can notice by looking at the optimization problem that it is not a semidefinite problem, nor that it has any other structure that is easy to exploit. For that reason finding larger instances is computationally very expensive.

We notice that Theorem~\ref{lem:counterexample} actually disproves a slightly stronger version of Conjecture~\ref{formal-conjecture-notmod}.
In the use of Conjecture~\ref{formal-conjecture-notmod} in \cite{carstens2018quantum}, they implicitly assume that we can replace $P_i$ by $\mathbbm{1}-P_i$ and the permutation still holds. While this modification gives a slightly stronger assumption, our counterexample in Theorem~\ref{lem:counterexample} works just as well. 

For the joint distribution in Conjecture~\ref{formal-conjecture-notmod} to exist, we know from Theorem~\ref{thm:joint-perms} that all operators must on-state permute. An important observation is that Conjecture~\ref{formal-conjecture-notmod} regards vector equalities and Theorem~\ref{thm:joint-perms} regards measurement outcomes. The theorem is ``easier'' than the former and our counterexample works with vector equalities, hence we indeed disprove Conjecture~\ref{formal-conjecture-notmod}.

In \cite{ebrahimi2018post} the author uses a different conjecture and different reasoning to prove the existence of the joint distribution. Our counterexample does not disprove this other approach and we refer interested readers to \cite{ebrahimi2018post}.

\section{Almost On-State Commutation}\label{sec:almost-commuting}
In the last part of the paper we discuss almost commutativity in the on-state case. In particular, we show here that if we have two projectors that almost commute on a state then we can define a projector that fully commutes with one of the original operators and is on-state close  to the second one. 

The main tool that we need to prove this result is the Jordan's lemma.
\begin{lemma}[Jordan's lemma \cite{jordan1875essai}]\label{lem:jordan}
	Let $ P_1 $ and $ P_2 $ be two projectors with rank $ r_i:=\textnormal{rank}(P_i)$ for $ i\in\{1,2\} $. Then both projectors can be decomposed simultaneously in the form $ P_i=\bigoplus_{k=1}^{r_i}P_i^k $, where $ P_i^k $ denote rank-1 projectors acting on one- or two-dimensional subspaces. We denote the one- and two-dimensional subspaces by $ S_1,  \dots,S_l$ and subspaces by $ T_1,  \dots,T_{l'}$, respectively.
	The eigenvectors $ \ket{v_{k,1}} $ and $ \ket{v_{k,2}} $ of $ P_1^k $ and $ P_2^k $ respectively are related by:
	\begin{align}
		\ket{v_{k,2}}=\cos\theta_k\ket{v_{k,1}}+\sin\theta_k\ket{v^{\perp}_{k,1}},
		\ket{v_{k,1}}=\cos\theta_k\ket{v_{k,2}}-\sin\theta_k\ket{v^{\perp}_{k,2}}.
	\end{align}
\end{lemma}

We can now prove our result.
\begin{thm}[Making almost commuting projectors commute]\label{lem:commute}
	Given any two projectors $P_1$ and $P_2$ and a state $\ket{\psi}$ we have that if $\norm{(P_1P_2-P_2P_1)\ket{\psi}}=\epsilon$ then there is a projector $P_2'$ that is close to the original projector on the state $\norm{(P_2'-P_2)\ket{\psi}}\leq\sqrt{2}\epsilon$ and $[P_1,P_2']=0$.
\end{thm}

\begin{proof}
    By Jordan's lemma (\Cref{lem:jordan}), there exist bits $\lambda_{i,1},\lambda_{i,2} \in \{0,1\}$ and vectors $\ket{u_1},...,\ket{u_m}$ and $\ket{v_{1,1}},\ket{v_{1,2}},...,\ket{v_{\ell,1}},\ket{v_{\ell,2}}$, such that
	\begin{enumerate}
	    \item $P_1 = \sum_{i \in [m]} \lambda_{i,1}\kb{u_i} + \sum_{i \in [\ell]} \kb{v_{i,1}}$ and $P_2 = \sum_{i \in [m]} \lambda_{i,2}\kb{u_i} + \sum_{i \in [\ell]} \kb{v_{i,2}}$;
	    \item $\braket{u_i}{u_k} = 0$ and $\braket{u_i}{v_{j,b}} = 0$ for all $b,i,j$ and $k \ne i$; 
	    \item $\braket{v_{j,b'}}{v_{i,b}} = 0$ for $i \ne j$ and any $b,b'$;
	    \item $0 < \braket{v_{i,1}}{v_{i,2}} < 1$.
	\end{enumerate}
	
		Let $\theta_i$ be the angle between $\ket{v_{i,1}}$ and $\ket{v_{i,2}}$ (i.e.\ $\cos \theta_i = \braket{v_{i,1}}{v_{i,2}}$), and  $\ket{v_{i,1}^\perp}$ be the state orthogonal to $\ket{v_{i,1}}$ in the subspace spanned by these two vectors. Since the non-commuting part of $P_1$ and $P_2$ must come from the pairs
	$\ket{v_{i,1}},\ket{v_{i,2}}$, we
    will define 
	$P_2'$ by removing the non-commuting part of $P_2$, shifting the vector $\ket{v_{i,2}}$, to either $\ket{v_{i,1}}$ or $\ket{v_{i,1}^\perp}$:
	\[
	P_2' =  \sum_{i \in [m]} \lambda_{i,2}\kb{u_i} +
	\sum_{i \in [\ell] : \theta_i \leq \frac{\pi}{4}}
	\kb{v_{i,1}}
	+ \sum_{i \in [\ell] : \theta_i > \frac{\pi}{4}}
	\kb{v_{i,1}^\perp}.
	\]
	
	We have clearly that $[P_1, P_2'] = 0$ since the two projectors are simultaneously diagonalizable and we now want to prove that
	\begin{align*}
	\norm{(P_2' - P_2)\ket{\psi}} \leq \sqrt{2}\eps.
	\end{align*}

	Notice that
	\begin{align}
	&\norm{(P_2' - P_2)\ket{\psi}}^2 \nonumber\\
	&= 
	\norm{\sum_{i \in [l'] : \theta_i \leq \frac{\pi}{4}}
	\left(\ket{v_{i,1}}\braket{v_{i,1}}{\psi} - \ket{v_{i,2}}\braket{v_{i,2}}{\psi} \right)
	+ \sum_{i \in [l'] : \theta_i > \frac{\pi}{4}}
	\left(\ket{v_{i,1}^\perp}\braket{v_{i,1}^\perp}{\psi} - \ket{v_{i,2}}\braket{v_{i,2}}{\psi} \right)}^2 \\
	&=	\sum_{i \in [l'] : \theta_i \leq \frac{\pi}{4}}
	\norm{\ket{v_{i,1}}\braket{v_{i,1}}{\psi} - \ket{v_{i,2}}\braket{v_{i,2}}{\psi}}^2
	+ \sum_{i \in [l'] : \theta_i > \frac{\pi}{4}}
	\norm{\ket{v_{i,1}^\perp}\braket{v_{i,1}^\perp}{\psi} - \ket{v_{i,2}}\braket{v_{i,2}}{\psi}}^2,
	\end{align}
	where in the last step we used that $\braket{v_{i,b'}}{v_{j,b}} = 0$ for $i \ne j$.
	
	Using that 
$\ket{v_{i,2}} = \cos{\theta_i}\ket{v_{i,1}} + \sin{\theta_i}\ket{v_{i,1}^\perp}$, we have that if $\theta_i \leq \frac{\pi}{4}$, then
	\begin{align}
	&\norm{\braket{v_{i,1}}{\psi}\ket{v_{i,1}} - \braket{v_{i,2}}{\psi}\ket{v_{i,2}}}^2 \nonumber \\ 
	&= \sin^4 \theta_i |\braket{v_{i,1}}{\psi}|^2 
	-  2\sin^3 \theta_i \cos \theta_i  \mathfrak{Re}(\braket{v_{i,1}^\perp}{\psi}\braket{v_{i,1}}{\psi}  )
	+  \sin^2 \theta_i \cos^2 \theta_i  |\braket{v_{i,1}^\perp}{\psi}|^2   \nonumber \\
	&+ \sin^4\theta_i  |\braket{v_{i,1}^\perp}{\psi}|^2 
	+ 2\sin^3 \theta_i \cos \theta_i  \mathfrak{Re}(\braket{v_{i,1}}{\psi}\braket{v_{i,1}^\perp}{\psi} )
	+  \sin^2 \theta_i \cos^2 \theta_i  |\braket{v_{i,1}}{\psi}|^2 \nonumber \\
	& \leq 2\sin^2\theta_i \cos^2 \theta_i (|\braket{v_{i,1}^\perp}{\psi}|^2 +  |\braket{v_{i,1}}{\psi}|^2), \label{eq:final1}
	\end{align}
	where in the inequality we used our assumption that $\theta_i \leq \frac{\pi}{4}$ which implies that $\sin \theta_i \leq \cos \theta_i$.  
	
	Using similar calculations, we have that if  $\theta_i \geq \frac{\pi}{4}$
		\begin{align}
	\norm{\braket{v_{i,1}^\perp}{\psi}\ket{v_{i,1}^\perp} - \braket{v_{i,2}}{\psi}\ket{v_{i,2}}}^2  \leq 2\sin^2\theta_i \cos^2 \theta_i (|\braket{v_{i,1}^\perp}{\psi}|^2 +  |\braket{v_{i,1}}{\psi}|^2). \label{eq:final2}
	\end{align}
	
We will show now that 
\[
 \sum_{i} \sin^2 \theta_i\cos^2 \theta_i (|\braket{v_{i,1}^\perp}{\psi}|^2 +  |\braket{v_{i,1}}{\psi}|^2)  = \eps^2,\] 
 which finishes the proof:
	\begin{align*}
	\eps^2 & =  \norm{(P_2P_1-P_1P_2)\ket{\psi}}^2\\
	&= \norm{
		\sum_{i \in [l']} \ket{v_{i,1}}\braket{v_{i,1}}{v_{i,2}}\braket{v_{i,2}}{\psi}  - \ket{v_{i,2}}\braket{v_{i,2}}{v_{i,1}}\braket{v_{i,1}}{\psi}.
	}^2 \\
	&=\norm{\sum_{i \in [l']} \cos{\theta_i} \left(\left(
	\cos{\theta_i}\braket{v_{i,1}}{\psi}  + \sin{\theta_i}\braket{v_{i,1}^\perp}{\psi} 
	\right)\ket{v_{i,1}}  - \braket{v_{i,1}}{\psi} \left( \cos{\theta_i}\ket{v_{i,1}} + \sin{\theta_i}\ket{v_{i,1}^\perp}\right)  \right)}^2\\
	&= \norm{\sum_{i \in [l']} \sin{\theta_i}\cos{\theta_i} \left(
	\braket{v_{i,1}^\perp}{\psi}\ket{v_{i,1}}  - \braket{v_{i,1}}{\psi} \ket{v_{i,1}^\perp}  \right)}^2 \\
	&= \sum_{i \in [l']} \sin^2 \theta_i \cos^ 2\theta_i \left(
	|\braket{v_{i,1}^\perp}{\psi}|^2 + |\braket{v_{i,1}}{\psi}|^2  \right).
	\end{align*}
	where in the second equality we again use that 
$\ket{v_{i,2}} = \cos{\theta_i}\ket{v_{i,1}} + \sin{\theta_i}\ket{v_{i,1}^\perp}$ and in the fourth equality we use the fact that $\braket{v_{i,b'}}{v_{j,b}} = 0$ for $i \ne j$.
\end{proof}

Our proof relies solely on Jordan's Lemma Note that Jordan's Lemma is sufficient only if we analyze commutation of projectors. Results that show how to make any Hermitian matrices commute \cite{friis1996almost,hastings2009making} are much more complicated to prove and it is not clear how to translate them to the ``on-state'' case.

We stress that our proof only works for two projectors, since Jordan's lemma does not generalize for three or more projectors. Therefore, we leave as open problem (dis)proving a generalized version of \Cref{lem:commute} for more projectors. %\new{We tried to tackle the case of more than just two operators but the main ingredient of our proof fails there. As far as we know, there is no version of the Jordan's lemma applicable to more than two projectors; For that reason our approach cannot be extended. }

In \cite{ebrahimi2018post} the author proves Theorem~\ref{lem:commute} for $ \eps=0 $, but with a different proof. They use Halmo's two rojections theorem instead of Jordan's lemma.

\bibliographystyle{alpha}
\bibliography{joints}

\newpage

\begin{appendix}
\section{Numerical values\label{sec:appendix-num}}
Below we present the state and the projectors that are claimed in the proof of Theorem~\ref{lem:counterexample}. The script used to generate these values can be found in \cite{joints-counterexample}. Before we write out the state and the projections that we found, let us state our violation of the permutator:
\begin{align}
	\|(P_1P_2P_3P_4-P_3P_4P_1P_2)\ket{\phi}\|=0.25 \pm 3\cdot 10^{-8}.
\end{align}
All the constraints listed in the proof of Theorem~\ref{lem:counterexample} are fulfilled up to the seventh decimal digit of precision, so up to $ 10^{-7} $. Internal computations of the algorithm are performed with machine precision of $ 10^{-15} $.

The state is
\begin{align}
\ket{\phi}:=\left( \begin{array}{c}
-0.135381-0.0503468 \text{i}\\0.325588\, -0.222403 \text{i}\\-0.209447-0.0404665 \text{i}\\-0.418336+0.130098 \text{i}\\-0.503693-0.299414 \text{i}\\0.379842\, +0.205081 \text{i}\\-0.179291-0.0381456 \text{i}\\0.0840381\, -0.125995 \text{i}
\end{array}\right).
\end{align}

We define the projectors by their eigenvectors:
\begin{align}
&P_1=\proj{\pi^1}, &P_2=\proj{\pi^2_1}+\proj{\pi^2_2}, \\
&P_3=\proj{\pi^3_1}+\proj{\pi^3_2}+\proj{\pi^3_3}, &P_4=\proj{\pi^4_1}+\proj{\pi^4_2}.
\end{align}
The eigenvector of $P_1$ is:
\begin{align}
\ket{\pi^1}:= \left( \begin{array}{c}
0.440777\, +0.168408  \text{i}\\
0.208781\, -0.37351 
 \text{i}\\
0.247514\, +0.0276065  \text{i}\\
-0.297971+0.0252308  \text{i}\\
0.118798\, +0.112225  \text{i}\\
-0.293428+0.270889  \text{i}\\
-0.193073+0.218869  \text{i}\\
-0.41405
\end{array}\right).
\end{align}
The eigenvectors of $P_2$ are:
\begin{align}
 \ket{\pi^2_1}:=\left( \begin{array}{c}
-0.497016-0.094035  \text{i}\\0.417527\, -0.0737062  \text{i}\\-0.000125303+0.35123  \text{i}\\0.166569\, -0.187245  \text{i}\\-0.373202+0.205633  \text{i}\\0.318452\, -0.251475  \text{i}\\-0.107473-0.123987 i\\-0.0711523
\end{array}\right),
\ket{\pi^2_2}:=\left( \begin{array}{c}
0.365906\, +0.0620997 \text{i}\\0.418728\, -0.2059 \text{i}\\0.229457\, +0.0557421 \text{i}\\-0.140393+0.0945029 \text{i}\\-0.199205-0.188139 \text{i}\\0.103617\, +0.279644 \text{i}\\-0.546498+0.147197 \text{i}\\0.275295
\end{array}\right).
\end{align}
The eigenvectors of $P_3$ are:
\begin{align}
& \ket{\pi^3_1}:=\left( \begin{array}{c}
-0.453059+0.181543 \text{i}\\-0.452841+0.0154095 \text{i}\\-0.17948-0.222827 \text{i}\\-0.230355-0.0526756 \text{i}\\-0.0918752-0.250754 \text{i}\\0.242416\, -0.126917 \text{i}\\0.300832\, -0.287566 \text{i}\\0.315259
\end{array}\right),
\ket{\pi^3_2}:=\left( \begin{array}{c}
-0.0586669-0.269559 \text{i}\\-0.280155+0.373271 \text{i}\\-0.150758-0.158539 \text{i}\\0.158793\, -0.0454731 \text{i}\\0.165888\, +0.362832 \text{i}\\-0.110453-0.310755 \text{i}\\0.353894\, -0.00811586 \text{i}\\-0.487537
\end{array}\right),\\
& \ket{\pi^3_3}:=\left( \begin{array}{c}
-0.182739-0.114718 \text{i}\\0.246775\, -0.134678 \text{i}\\-0.513357-0.193655 \text{i}\\-0.10451+0.421294 \text{i}\\0.111183\, +0.122625 \text{i}\\-0.200917-0.25897 \text{i}\\-0.0290851+0.398494 \text{i}\\0.30081
\end{array}\right).
\end{align}
The eigenvectors of $P_4$ are:
\begin{align}
\ket{\pi^4_1}:=\left( \begin{array}{c}
-0.464187+0.213035 \text{i}\\-0.364421+0.119836 \text{i}\\-0.324984-0.23097 \text{i}\\-0.256841+0.0478513 \text{i}\\-0.0700499-0.192822 \text{i}\\0.146148\, -0.225755 \text{i}\\0.243944\, -0.284786 \text{i}\\0.331272
\end{array}\right),
\ket{\pi^4_2}:=\left( \begin{array}{c}
0.111757\, +0.151275 \text{i}\\0.236223\, -0.323279 \text{i}\\0.157312\, -0.115385 \text{i}\\-0.30864+0.0990552 \text{i}\\-0.260931-0.236239 \text{i}\\0.240497\, +0.13559 \text{i}\\-0.453404+0.12357 \text{i}\\0.490125
\end{array}\right).
\end{align}

\end{appendix}
\end{document}